\documentclass[10pt,journal,cspaper,compsoc]{IEEEtran}
\usepackage{amsmath,amsthm}
\usepackage{authblk}
\usepackage{algorithmic}
\usepackage{algorithm}
\usepackage{stmaryrd}
\usepackage{pxfonts}
\usepackage{graphicx}
\usepackage{array}
\usepackage{flushend}

\begin{document}
\title{Reputation Aggregation in Peer-to-Peer Network Using Differential Gossip Algorithm}

\author{Ruchir~Gupta, Yatindra~Nath~Singh,~\IEEEmembership{Senior Member,~IEEE,}

\thanks{Ruchir Gupta and Yatindra Nath Singh are with the Department of Electrical Engineering,
Indian Institute of Technology, Kanpur, India.}%
\thanks{E-mail:{\{rgupta, ynsingh}\}@iitk.ac.in}}

\newtheorem{theorem}{Theorem}[section]
\newtheorem{lemma}[theorem]{Lemma}
\newtheorem{proposition}[theorem]{Proposition}
\newtheorem{corollary}[theorem]{Corollary}

\IEEEcompsoctitleabstractindextext{%
\begin{abstract}
Reputation aggregation in peer to peer networks is generally a very time and resource consuming process. Moreover, most of the methods consider that a node will have same reputation with all the nodes in the network, which is not true. This paper proposes a reputation aggregation algorithm that uses a variant of gossip algorithm called differential gossip. In this paper, estimate of reputation is considered to be having two parts, one common component which is same with every node, and the other one is information received from immediate neighbours based on the neighbours' direct interaction with the node. The differential gossip is fast and requires less amount of resources. This mechanism allows computation of independent reputation value by a node, of every other node in the network, for each node. The differential gossip trust has been investigated for a power law network formed using preferential attachment \emph{(PA)} Model. The reputation computed using differential gossip trust shows good amount of immunity to the collusion. We have verified the performance of the algorithm on the power law networks of different sizes ranging from 100 nodes to 50,000 nodes.
\end {abstract}
}

\maketitle
\IEEEdisplaynotcompsoctitleabstractindextext
\IEEEpeerreviewmaketitle
\section{Introduction}
\IEEEPARstart{P}{eer}-to-peer systems have attracted considerable attention in recent past as these systems are more scalable than the client-server systems. But, free riding has emerged as a big challenge for peer-to-peer systems \cite{Feldman:2005:OFB:1120717.1120723,10.1109/MIC.2009.33}. Tendency of nodes to draw resources from the network and not giving anything in return is termed as ’Free Riding’. As nodes have conflicting interests so, selfish behaviour of nodes leads to problem of free riding \cite{conflictinginterest}. This behaviour of nodes can be explained by famous prisoners' dilemma \cite{osborne}. In a file sharing network, if nodes are considered as players, their NE will be the strategy where none of them will share the resources \cite{Tang}. Experimental studies \cite{Adar00freeriding,Hughes2005} on Gnutella network confirmed this fact. 

Trust or reputation management systems can be used in peer to peer networks to overcome the problem of free riding as well as to ward off some of the attacks \cite{10.1109/TPDS.2008.60,1209003}. Such reputation management systems have been in e-commerce portals like e-bay \cite{ebay}, but they have the advantage that they are based on client server architecture. In peer-to-peer file sharing networks, as there is no central server or repository, trust has to be estimated and stored by each node in a distributed fashion, and all such trust values need to be aggregated to build an effective reputation management system. Aggregation of trust generally consumes a lot of time and memory especially with large number of nodes. Apart from it, existing methods assume that the reputation of a peer must have a global value, i.e. the peers behave uniformly with all other peers, but this is not true. In this paper, we propose a method which aggregates the trust estimated by a node directly, trust reported by neighbours and trust averaged using a variation of gossip algorithm to make trust vectors at each node. Studies show that unstructured peer-to-peer networks generally follow \emph{Preferential Attachment} (PA) 
model. In particular Gnutella has a power law degree distribution $f(d) = d^{-\alpha}$ with $\alpha =2.3$ \cite{Powerlaw, Powerlaw1, Powerlaw2}. Therefore this algorithm has been simulated for networks having power law degree distribution, i.e. networks formed by Preferential Attachment (PA) model \cite{Barabasi, PreA}. According to \cite{PreA}, a graph $G^{m}_{N}$ evolves from $G^{m}_{N-1}$ when a new node with $m$ edges joins the network. Here $m$ is the number of connections that the new node will make at the time of joining the network. The joining nodes chooses a node $i$ with probability $P_{i}$ given by 
$P_{i}=\frac{degree\; of\; node\; i\; before \; this \; connection \; is\; made}{sum\; of\; degree\; of\; all\; the\; nodes\; before\; this\; connection\; is\; made}$.  

Remainder of this paper is organised as follows. Section two presents the related work in reputation management. Section three describes the system model. Section four, proposes the differential gossip trust algorithms for the aggregation of trust. Section five presents the analysis of the algorithms and numerical results. Section six concludes the presented work.

\section{Related Work}
Many methods have been proposed in the literature \cite{eigentrust,PEERTRUST,FuzzyTrust,powertrust,gossiptrust, pet} for the reputation aggregation. Eigen-Trust \cite{eigentrust} depends largely on pre-trusted peers i.e. peers that are globally trusted. This is scalable to a limited extent. Peer-Trust \cite{PEERTRUST} stores the trust data (i.e. trust values of all the peers in the network) in a distributed fashion. This is performed using a trust manager at every node. In Peer-Trust, hash value of a node id is calculated to identify the peer where the trust value of the node will be stored. Song \emph{et.al.} \cite{FuzzyTrust} used fuzzy inference to compute the aggregation weights. These weights are assigned on the basis of global trust value of the opining peer, transaction amount and date of transaction with the peers being evaluated. In Fuzzy-Trust, each peer maintains a local trust value and the transaction history with the remote peers. At the time of aggregation, system queries for the trust value, from the qualified peers, for the peer being evaluated, and combine the received values using aggregation weights to compute the updated trust values. The trust value calculated in this way is global in nature. Zhou \emph{et.al.} \cite{powertrust} identified the power law distribution in the users' feedback by observing eBay transaction traces. Power Trust leverage on this distribution. Every node keeps the record of local trust scores of the nodes with whom it has interacted, as per the quality of service obtained from them. Global reputation is calculated by weighted aggregation of local trust scores. The global reputation of opining node is taken as weight. Similar to fuzzy-trust, trust value generated here are global and all peers in the network use the same value for all the transactions. Gossip Trust \cite{gossiptrust} uses push gossip algorithm as given in \cite{Kempe} for aggregation in complete graph based network. This ensures fast reputation aggregation with low message overhead. This technique also calculates the global reputation of a node by weighted aggregation of opinions of different nodes. It also uses the bloom filter architecture for the efficient ranking. \cite{Kempe} analyses gossip algorithm for aggregation of information in complete graph based network and also comes up with a upper bound on convergence. This paper also studies diffusion in the network.

As evident, generally the earlier work \cite{eigentrust,PEERTRUST,powertrust,gossiptrust} assumes that the reputation of a peer must have a global value, i.e. peers behave uniformly with all the peers. But this is not the case due to nodes being selfish in nature. A peer behaves with different decency levels with different peers. The same holds true for the opinion as well, i.e. peer gives different weights to opinions of different peers. These aspects have been taken into account in our algorithm called Differential Gossip Trust. The current work combines local, reported value from trusted neighbours and global trust using the new variation of gossip to define a novel trust aggregation algorithm.
\section{System Model}
In this paper, we are studying a peer-to-peer network. Typically, there will be millions of nodes in a peer-to-peer network. These nodes are assumed to be connected by a network graph $G_N ^m$ generated by PA (preferential attachment) process \cite{Powerlaw} for $m\geq 2$. Here, all the nodes whose addresses are stored by a node, are considered to be neighbours of the node. Generally the nodes will have small number of neighbours.

There is no dedicated server in this network. Peers in this network are rational, i.e. they are only interested in their own welfare. They are connected to each other by an access link followed by a back bone link and then again by an access link to the second node. We are assuming that the network is heavily loaded i.e. every peer has sufficient number of pending download requests, hence these peers are contending for the available transmission capacity. We also assume that every peer is paying the cost of access link as per the use (for both download and upload as per the billing practice of most of the service providers). Downloaded data is more valuable than the cost of access link. Moreover, data that is of interest to peer is always available. So, every peer wants to maximises its downloads and minimise its uploads so that it can get maximum utility of its spending. This optimisation leads to problem of free riding.

If a node is downloading, some other node has to upload. So the desired condition is that the download should be equal to upload for a node. Usually this means that there is no gain. Even in this scenario, the node gains due to interaction with others, as the chances of survival of any entity is more with communication capability. Thus interaction itself is an incentive. A node will usually try to get the content and avoid uploading to maximise gain. Thus free riding becomes optimal strategy. So a reputation management system need to be enforced to safeguard the interest of every node by controlling the free riding behaviour.

In a reputation management system, every node maintains a reputation table. In this table, a node maintains the reputation of the nodes with whom it has interacted. Whenever it receives a resource from some node, it adjusts the reputation of that node accordingly. When another node asks for the resource from this node, it checks the reputation table and according to the reputation value of the requesting node, it allocates resource to the other node. This ensures that every node is facilitated from the network as per its contribution to the network and consequently free riding is discouraged.

For using such a reputation management system, the nodes need to estimate the trust value of the nodes interacting with them. There are number of ways to estimate the reputation \cite{ruchir2}. We assume that trust value observed by node a $i$ for the node $j$ can be defined as $t_{ij}$.

%
\section{Aggregation of Trust}
 Whenever a node needs a resource, it asks from its neighbours; if they have the resource, the node gets the answer of its query. If neighbours do not have it, they forward the query to their neighbours and so on. The node that have the resource, replies back to the requesting node. The requesting node now asks for the resource from the node having the resource. The answering node provides the resource now directly  according to the reputation of the node. 

If a node receives a request from another node that is not its neighbour, the reputation of that node needs to be estimated some how in order to decide the quality of service to be provided. If two nodes are going to transact for the first time they should have reputation of each other. This can be done by getting the reputation of node from neighbours and then using it to make an initial estimate. When for a node, multiple trust values are received, we need an aggregation mechanism to get the trust value. Trust value should always lie in between zero and one.

For the whole network, we can define a trust matrix of dimensions $N\times N$. Here $t_{ij}$ represents the trust value of $j$ as maintained by $i$ based on direct interaction. This matrix is generally sparse in nature as generally a node
will have very small number of neighbours being directly transacted with as compared to total number of nodes in the network. It may be noted that $t_{ij}$ is estimated based on transaction between nodes $i$ and $j$ and can be called as local trust value.  These trust values will be propagated and aggregated by all the nodes in a network. The trust estimate which should be actually used will be based on aggregation of local trust values and trust estimates received from neighbours. For the reputation information received from direct neighbours, the weights can be assigned based on neighbours' reputation.

\subsection{Differential Gossip Trust}
We have modified the gossip based information diffusion algorithm to  allow faster diffusion of the trust values enabling faster estimation of global trust vectors at all the nodes. The algorithm can be divided into two parts. In first part, we will discuss about the method of information diffusion whereas in the second part, we will discuss about the information that is to be diffused.
\subsubsection{Differential Gossip Algorithm}
Gossip Algorithms are used for spreading information in large decentralised
networks. These algorithms are random in nature as in these, nodes randomly choose their communication partner in each information diffusion step. These algorithms are generally simple, light weight and robust for errors. They have less overhead compared to the deterministic algorithms \cite{Kempe}, \cite{Shah}. The gossip algorithms are also used for the distributed computation like taking average of the numbers stored at different nodes. These algorithms are suitable for the computation of reputation vector in the peer-to-peer networks \cite{gossiptrust}. 

There are three types of gossip algorithms: push, pull and push-pull. In push
kind of algorithms, in every gossip step, nodes randomly choose a node from its neighbours and push a piece of information to it. Whereas in pull algorithms, nodes take the information from one of the randomly selected neighbouring node. Both these processes happen simultaneously in the push-pull based algorithms.


Chierichetti \emph{et.al.} \cite{Flavio} stated that in a PA model based network, push or pull alone will take long in spreading the information in the network. If the push model is implemented and the information is with a power node, it will take many rounds in pushing information to low degree nodes. If pull model is being used, and the information is with low degree node, it will again take many rounds for a power nodes to pull the information. This phenomenon will be evident in the average computation using push or pull gossip algorithm as information of every node need to be distributed to every other node.

To avoid this problem we propose differential push gossip algorithm. In this algorithm, every node makes different number of pushes depending upon the ratio of  its own degree to the average neighbour degree. in a single gossiping step.  If every node also pushes its degree to all the neighbouring nodes, then each node can estimate the average degree of all its neighbours. Here, we make three assumptions, i) every node has a unique identification number known to every other node. So, if some node pushes some information about another node, receiving node knows that this information is about which particular node; ii) time is discrete; and iii) every node knows about the starting time of gossip process.

All nodes that have some feedback about a single node, gossip their feedback about that single node. All the nodes estimate the global reputation of the single node based on the outcome of gossip. Let the feedback about the $j^{th}$ node by node $i$ be $y_{ij}$. If it does not have any feed back about $j$ it keeps the value of $y_{ij}$ as $0$. Every node that has feed back about node $j$ assumes the gossip weight $g_{ij}$ as $1$ and rest of the nodes assume  the gossip weight as zero. It is done so that as a result of gossip, every node coverages to the ratio of summation of all gossiped values as well as summation of all gossip weights averaged over all the nodes $N$. The ratio of two converged values , i.e. $\frac{\sum\nolimits_{i}y_{ij}}{\sum\nolimits_{i}g_{ij}}$ gives the sum of gossiped values averaged over node who started with gossip weight of unity. If only one of the nodes assumes gossip weight as unity and remaining as zero, all nodes will be able to estimate the summation of gossiped values.

Every node has $y_{ij}$ and $g_{ij}$ as information to be gossiped. Let us call this pair as gossip pair. The ratio of  gossip pair is tracked in every step to decide on convergence. Every node, first calculates the ratio $k$ of its own degree and average degree of its neighbours. As $k$ will be a real number, it is rounded off to nearest integer if $k\ge 1$. For all other cases, $k=1$. The node chooses k nodes randomly in its neighbourhood and sends ($\frac{1}{k+1}y_{ij},\frac{1}{k+1}g_{ij}$) as gossip pair to all randomly selected $k$ nodes and itself. The node is also considered as one of the neighbours of itself.

After receiving all gossip pairs from different nodes including itself, the node sums up all the pairs. This summation now becomes new gossip pair. The ratio of this gossip pair is the value that node has evolved in this step. If at least one gossip pair has been received from a node other than itself, the condition of convergence will be checked between the ratios of this step and previous step with a predefined error constant. If convergence condition is satisfied, it means a node needs not to run the gossip process any more for its convergence. But once convergence is achieved by a particular node, convergence of other nodes is not assured. Hence, if a node will stop gossiping, convergence of its neighbours may suffer. To avoid this problem, once a node gets converged, it will announce among all its neighbours that it has achieved convergence. Every neighbour will note this announcement. When a node finds that itself and all its neighbours have converged, it will stop the gossip process.

When a round of gossiping starts it takes some time to complete. After gossiping, nodes get a value that is used till the next new value is converged upon at the end of next round. After the end of a round, next round of gossip will start after some time. The time difference between the two rounds will depend upon the change in the behaviour of the nodes in the network and the number of new nodes coming in the network per unit time. For simplicity, this time difference has been taken as a constant. In reality, this should be dynamically adjusted.
\subsubsection{Differential Reputation Aggregation}
When a node $j$ requests resource from another node $i$, the node $i$ needs the reputation of node $j$ so that it can decide the quality of service to be offered to node $j$. There are two possible conditions between node $i$ and node $j$. First, node $j$ may have served node $i$ earlier and hence node $i$ has some trust value about node $j$. In this case there is no problem for node $i$ and node $i$ will serve as per the reputation available with it. Second, node $i$ and node $j$ are unknown to each other. In this case node $i$ needs general reputation about node $j$. For this, aggregation of reputation is needed. 

Aggregation of reputation should not be resource intensive and should be immune to collusion and whitewashing. We can have two kind of options for this. First, by gossiping all the nodes can reach a consensus about the reputation of node $j$ \cite{gossiptrust}. Second all the nodes exchange their reputation tables about the nodes they have interacted with and this process should always be running, like executing a routing protocol at network layer. First process is more vulnerable to collusion where as second process is resource intensive. As we can see that unstructured  peer-to-peer network is very similar  to human network, thus we can observe the human behavioural strategies to identify the solution. In human network when we need the reputation value of some body we rely on personal experience with him. If we don't have any personal experience with him, we rely on two things. First the general perception about him which we receive from gossip flowing around and second the information given by our friends if they have any direct interaction with him. We combine these two and act accordingly. 

Nodes follow the same kind of approach in our proposed algorithm. The nodes gather opinion of their neighbours and combine it with the opinion, obtained from general gossip after weighing the neighbours' opinion according to the confidence in the neighbours. In general, it can be said that a node gives weight to every node in the network. The nodes that have not interacted with it are given weight as $1$ where as those which have interacted are given weight according to the confidence in them (always $\ge 1$).

Let us consider that there are $N$ nodes in the network. Every node periodically calculates the trust value of the other nodes on the basis of quality of service provided by them against the requests made. Let us assume that $t_{ij}$ is the trust value measured by node $i$ for node $j$. Here $t_{ij}$ ($1\leq i,j \leq N$) will always lie between 0 and 1 such that the $t_{ij}=1$ will represent the complete trust in node $j$, whereas, $t_{ij}=0$ will represent no trust in the node $j$.

If a node 'A' has not transacted with a node 'B', then the trust value of node 'B' will also remain $0$ with the node 'A'. This initial value is taken as $0$ to avoid the white washing attack. This initial value can also be taken as higher than zero and can be dynamically adjusted thereafter as per the level of whitewashing in the network. In this paper, we have not studied this aspect.

We will discuss the algorithm in four steps to make it simple to understand. In the first step, global reputation aggregation for a single node will be discussed. In the second step, we will discuss globally calibrated local reputation aggregation for this single node. In the third step, simultaneous global reputation aggregation for all the nodes will be discussed, and finally in the fourth step, simultaneous aggregation of globally calibrated local reputation aggregation for all the nodes will be discussed.
 
 In the first step to keep things the simple, we assume that weights of all the nodes for every node be $1$. This leads us to calculation of global reputation ($\bf R_{global}$) of a node. This can be equivalently represented using a matrix vector multiplication as follows,
 \begin{equation}
 \label{gr}
{\bf R_{global}}(n+1)=\frac{1}{N}({\bf t^{T}}(n)\times{\bf 1_{N \times 1}}).
 \end{equation}
 Here n is the time instant and ${\bf R_{global}}$ is the global reputation vector containing global reputations of different nodes. Let us assume that $R_{i}$ is global reputation of $i^{th}$ node i.e. $i^{th}$ element in $\bf R_{global}(n+1)$ column vector. ${\bf 1_{N \times 1}}$ is a vector of $N$ $1's$, i.e. $[1 1 1 1.......]^{T}$. Differential Gossip algorithm can be used for doing this computation in distributed fashion. This process is shown in algorithm~\ref{alg1}.
\begin{algorithm}
\caption{Global Reputation aggregation for a single node}
\label{alg1}
\begin{algorithmic}
\REQUIRE $t_{ij}$ (The reputation estimated by node $i$ for node $j$ only on the basis of direct interaction) for $1\leq i \leq N$, gossip error tolerance $\xiup$
\ENSURE Global Reputation of node $j$ ($R_{j}$)
\IF {$i$ has some reputation value about $j$}
\STATE Assume weight $g_{ij}=1$, and $y_{ij}=t_{ij}$
\ELSE
\STATE Assume weight $g_{ij}=0$, and $y_{ij}=0$
\ENDIF
\STATE Push self degree to neighbouring node
\STATE Take the average of neighbours' degree
\STATE Calculate the ratio of its degree and average of neighbour degree ($k_{i}\shortleftarrow
\frac{degree\; of\; i}{average\; neighbour\; degree}$)
\STATE Round off $k_{i}$ to nearest integer for  $k_{i}\geq 1$ else take $k_{i}=1$
\STATE $m\shortleftarrow 1$
\COMMENT {Initialise Gossip Step}
\STATE $u\shortleftarrow \frac{y_{ij}}{g_{ij}}$ for nodes having $g_{ij}\ne 0$; otherwise $u\shortleftarrow 10$.
\REPEAT
\STATE Do for node $i$ 
\STATE $(y_{sj},g_{sj})$ are all pairs (of gossip weight and gossip value) received by the node $i$ in the previous step
\STATE $y_{ij}\shortleftarrow \sum\limits_{s\in S}y_{sj};\;g_{ij}\shortleftarrow
\sum\limits_{s\in S}g_{sj}$
\COMMENT{update gossip pairs}
\COMMENT{S is the set of nodes sending the gossip to $i$}
\STATE choose $k_{i}$ random nodes in its neighbourhood
\STATE send gossip pair ($\frac{1}{k_{i}+1}y_{ij},\frac{1}{k_{i}+1}g_{ij}$) to all $k_{i}$
nodes and also to itself
\STATE $m\shortleftarrow m+1$
\COMMENT {increment the gossip step}
\IF {$|S|> 1$}
\IF {$|\frac{y_{ij}}{g_{ij}}-u|\leq \xiup$}
\STATE {Inform all neighbours about self convergence}
\ENDIF
\ENDIF 
\STATE $u\shortleftarrow \frac{y_{ij}}{g_{ij}}$ for nodes having $g_{ij}\ne 0$; otherwise $u\shortleftarrow 10$.
\UNTIL
 {Self convergence and all neighbours' convergence has happened}
\STATE {\bf output $R_j=\frac{y_{ij}}{g_{ij}}$}
\end{algorithmic}
\end{algorithm}

Although each node considers the average of feedbacks from every other node in the network, it is desirable to assign different weights to the direct feedbacks received from neighbouring nodes. The direct feedback from a node is based on the direct interaction which it had experienced. The weights can be assigned by a node on the basis of number and quality of transactions made with the other node who is providing the feedback. The trust value of a node is a good metric of quality and number of transactions. The weights for different nodes can be derived on the basis of the trust values of these nodes. Same idea is used in second variation of algorithm where nodes estimate globally calibrated local reputation vector. So in second step, we propose the weight $w_{ij}$ to be of the form:
\begin{equation}
\label{weight}
w_{ij}=a_i^{b_{ij}\cdot t_{i j}}.
\end{equation}
Here $a_i$ and $b_{ij}$ are two parameters that a node can decide on its own. First parameter can be adjusted according to the overall quality of service received by the node from the network, whereas second parameter can be adjusted according to the recommendation of a particular neighbour and quality of service from the network. So the second parameter will be adjusted for every neighbour independently. In this paper, $a_i$ and $b_{ij}$ has been taken as the constants for every node for simplicity. Salient features of this scheme are as follows.
\begin{itemize}
\item Even if a node has no neighbourhood relationship with the estimating node, its feed back will still get some consideration.
\item If a node has bad reputation with the estimating node, its feedback will have weight close to the node which have no neighbourhood relation with the estimating node.
\item Nodes with higher reputation will be given higher weights and it will help in making better quality of service groups.
\item Values of $a_i$ and $b_{ij}$ can be dynamically adjusted by nodes as per their requirement. Though in this work, $a_i$ and $b_{ij}$ have been taken as constants.
\item Collusion will be significantly reduced. 
\end{itemize}
A weighted trust matrix, that is different at every node, is formed by multiplying trust values with weights i.e. for node $I$, the $ij$ element in weighted matrix will be $W_{Iij}$ such that,
\begin{equation}
W_{Iij}=w_{Ii}\times t_{ij}.
\end{equation}
A node gives high weight to the feedback given by those nodes which have provided better quality of service. This leads to the calculation of globally calibrated local reputation vector. It is a collection of the reputation of all the nodes in the network according to received feedback about the node and the weights of nodes giving feedback to the calculating node. It means if some node I is calculating globally calibrated local reputation vector, the $j^{th}$ element of this vector will be
\begin{equation}
\label{weigh}
Rep_{I,j}=\frac{\sum\limits_{i}W_{Iij}}{\sum\limits_{i}w_{Ii}}.
\end{equation}
Now globally calibrated local reputation at node I, ${\bf R_{gclr\;I}}$ can be equivalently represented as the matrix vector multiplication
\begin{equation}
\label{gclr}
{\bf R_{gclr\;I}}(n+1)=\frac{1}{Sum_I}({\bf W_{I}^{T}}(n)\times{\bf 1_{N \times 1}}).
\end{equation}
Here $Sum_I=\sum\limits_{i}w_{Ii}$. It is interesting to see that if we consider the weights of all nodes as $1$ in (\ref{gclr}), this equation degenerates to (\ref{gr}).

Each node will have four different kind of data about other nodes - first $t_{ij}$ i.e. the trust value as result of direct interaction, second $y_{ij}$, the intermediate variable for gossiping and third $g_{ij}$, gossiping weight. In the start of every gossiping round  $y_{ij}$  assume the value of $t_{ij}$ whereas $g_{ij}$ will be $1$ only for one of the values of $i$, and $0$ for all the others. This will happen for all values of $i$ and $j$. Fourth value $Rep_{ij}$ is obtained after gossiping and consideration of neighbours' opinion. It will also be maintained at every node $i$. Apart from these four entities, we also wants to count the total number of nodes opining about node $j$. For this purpose every node that have opined about node $j$ will assume $count_{ij}=1$ and others will assume $count_{ij}=0$ and hence in the process of gossip, all these $1$'s will sum up and we will get the count.

Equation (\ref{weigh}) can be alternatively represented as,
\begin{eqnarray}
\nonumber Rep_{I,j}&=&\frac{\sum\limits_{i\in NS_I}W_{Iij}+\sum\limits_{i\notin NS_I}W_{Iij}}{\sum\limits_{i\in NS_I}w_{Ii}+\sum\limits_{i\notin NS_i}w_{Ii}},\\
\nonumber &=& \frac{\sum\limits_{i\in NS_I}w_{Ii}\times t_{ij}+\sum\limits_{i\notin NS_I}w_{Ii}\times t_{ij}}{\sum\limits_{i\in NS_I}w_{Ii}+\sum\limits_{i\notin NS_I}w_{Ii}},\\
\nonumber &=& \frac{\sum\limits_{i\in NS_I}(w_{Ii}-1)\times t_{ij}+\sum\limits_{i\notin NS_I}(w_{Ii}-1)\times t_{ij}+\sum\limits_{i} t_{ij}}{\sum\limits_{i\in NS_I}(w_{Ii}-1)+\sum\limits_{i\notin NS_I}(w_{Ii}-1)+\sum\limits_{i}1}.
\end{eqnarray}
Here $NS_I$ is the set of neighbours of node $I$. As neighbourhood between two nodes is based upon the interaction between them so for non neighbour nodes the weight will be $1$. Using this fact,
\begin{equation}
\label{algo2eqn}
Rep_{I,i}= \frac{\sum\limits_{i\in NS_I}(w_{Ii}-1)\times t_{ij}+\sum\limits_{i} t_{ij}}{\sum\limits_{i\in NS_I}(w_{Ii}-1)+\sum\limits_{i}1}.
\end{equation}

In order to compute the globally calibrated local reputation of node $j$, each node will need the reputation of $j$ as estimated by neighbours on the basis of the direct interaction with node $j$. Whereas in the gossip algorithm, after every step, the value at the node keeps on changing as it gets added to the values pushed by other nodes and is distributed after division to neighbours. After few steps, the values converge when incoming and outgoing values statistically balance each other. Hence after the first step of gossiping it is difficult to get the value of $t_{ij}$ from a neighbour for the nodes with whom it has direct interaction, by gossiping process. 

If a node is participating in the process of gossip about node $j$ for the first time  or reputation of node $j$ at this node has changed considerably since start of previous round of gossip, this node will inform the reputation of node $j$ to all of its neighbours before the start of next gossiping round (figure \ref{fig:gclr}). This will be done by all the nodes. After this process, every node has opinion of its neighbours about node $j$ (If a node does not inform reputation of $j$, the already available earlier value will be considered). If node will not hear from a node for a long time, it will assume that this node is no longer present and hence it will drop its feedback after some time.
 
Now these reputations will be multiplied by ($W_{Ii}-1$) as required in equation (\ref{algo2eqn}) and summed up as value $\hat{y}_{Ij}$ (see algorithm 2). Now normal gossip will be done as in algorithm~\ref{alg2} with a difference that only one node will be given gossip weight $1$ and rest will be given $0$ gossip weight The nodes that have reputation information about node under consideration (i.e., $j$) will also push $count =1$. After stabilisation of gossip, each node will have $\frac{(N_d)}{N}$, where $N_d$ is number of nodes having direct interaction and $N$ is total number of nodes, sum of total values/total number of nodes, and $\frac{1}{N}$ as gossip weight. This will lead to the summation of all reputation values available and total number of nodes giving these reputation values. Now reputation can be calculated using (\ref{algo2eqn})[algorithm~\ref{alg2}]. 
 
 \begin{figure}[!t]
 \begin{center}
 \includegraphics[width=80mm,height=46mm, keepaspectratio=false]{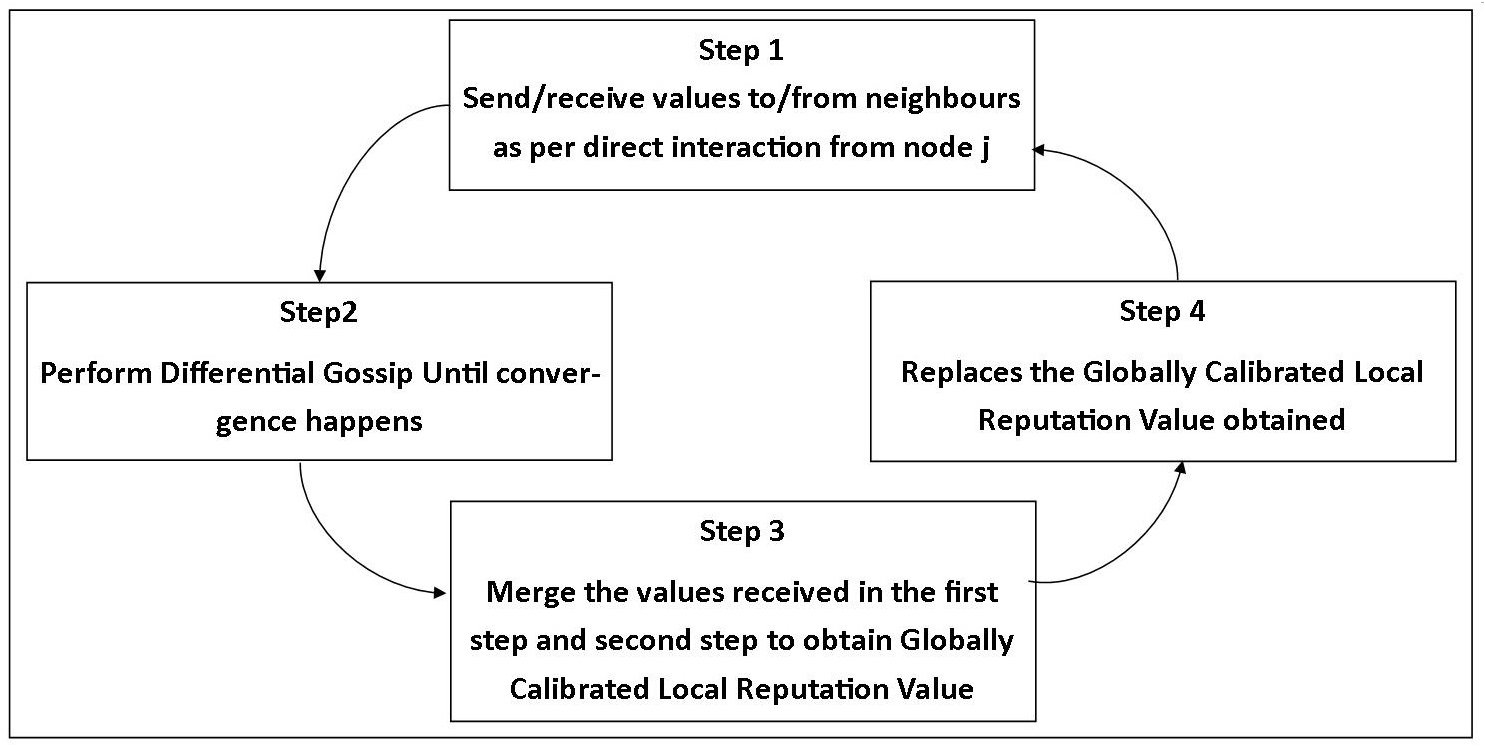}
 \caption{Sequence of computation for estimating globally calibrated local reputation at each node}
 \label{fig:gclr}
 \end{center}
 \end{figure}

\begin{algorithm}
\caption{Globally calibrated local Reputation aggregation for a single node}
\label{alg2}
\begin{algorithmic}

\REQUIRE Feedback matrix ${\bf t}$, gossip error tolerance $\xiup$. 
\ENSURE Globally calibrated local Reputation of node $j$ ($Rep_{ij}$)
\STATE Assume weight $g_{1}=1$
\STATE Node $i$ do
\IF {$i\ne 1$}
\STATE $g_{ij}=0$
\ENDIF
\IF {$i$ has some reputation value about $j$}
\STATE take $count_{ij}=1$, and $y_{ij}=t_{ij}$
\ELSE
\STATE take $count_{ij}=0$, and $y_{ij}=0$
\ENDIF
\STATE calculate $w_{j}$ for all neighbours of $i$ by formula
$w_{ij}=a_i^{b_{ij}t_{ij}}$
\IF {Node $i$ is participating first time in gossiping process}
\STATE Push feedback due to direct interaction about the node under consideration to all neighbours
\ELSE
\IF {Feedback about the node under consideration has changed by more than some constant $\Delta$ }
\STATE Push the new feedback to all the neighbours
\ENDIF
\ENDIF
\STATE Push self degree $d_{i}$ to neighbouring nodes
\STATE Calculate $\hat{y}_{ij} = \sum\limits_{k=1}^{d_i} ( w_k -1 ) \times feedback\; from\; node\; k\; about\; j$

\bf{Algorithm Continued...}
\end{algorithmic}
 \end{algorithm}

 \begin{algorithm}
\renewcommand{\thealgorithm}{}
 \nonumber\begin{algorithmic}
 \caption{Algorithm 2 (continued)}
 \STATE Take the average of neighbours' degree
 \STATE Calculate the ratio of its degree and average of neighbour degree ($k_{i}\shortleftarrow
 \frac{degree\; of\; i}{average\; neighbour\; degree}$)
 \STATE Round off $k_{i}$ for  $k_{i}\geq 1$ else take $k_{i}=1$
 \STATE End do
 \STATE $m\shortleftarrow 1$ 
 \COMMENT {Initialise Gossip Step}\\
 \STATE $u\shortleftarrow \frac{y_{ij}}{g_{ij}}$ for nodes having $g_{ij}\ne 0$; otherwise $u\shortleftarrow 10$.
\REPEAT
\STATE Do for node $i$ 
\STATE $(y_{sj},g_{sj},count_{sj})$ are all 3-tuples of gossip value of reputation, gossip weight and count received by node $i$ in the previous step

\STATE $y_{ij}\shortleftarrow \sum\limits_{s\in S}y_{sj};\;g_{ij}\shortleftarrow
\sum\limits_{s\in S}g_{sj};\;count_{ij}\shortleftarrow \sum\limits_{s\in S}count_{sj}$
\COMMENT{update gossip pairs}\\
\COMMENT{S is the set of nodes sending the gossip to $i$}
\STATE choose $k_{i}$ random nodes in its neighbourhood
\STATE send gossip pair ($\frac{1}{k_{i}+1}y_{ij},\frac{1}{k_{i}+1}g_{ij})$ to all $k_{i}$ nodes and itself
\STATE $m\shortleftarrow m+1$
\COMMENT {increment the gossip step}
\IF {$|S|> 1$}

\IF {$|\frac{y_{ij}}{g_{ij}}-u|\leq \xiup$}
\STATE {Inform all neighbours about self convergence}
\ENDIF
\ENDIF 
\STATE \STATE $u\shortleftarrow \frac{y_{ij}}{g_{ij}}$ for nodes having $g_{ij}\ne 0$, otherwise $u\shortleftarrow 10$.
\UNTIL
 {Self convergence and all neighbours' convergence has happened}
%
%
\STATE $ {\bf output}Rep_{ij}\shortleftarrow \frac{\hat y_{ij} +\frac{y_{ij}}{g_{ij}}}{\sum (w_{k}-1)+\frac{count_{ij}}{g_{ij}}}$
\end{algorithmic}
\end{algorithm}

In third variation we want to aggregate the global reputation of all nodes
simultaneously. This algorithm is quite similar to algorithm~\ref{alg1} except few
changes. Unlike algorithm~\ref{alg1}, node will push complete vector $\bf{y_i}$ which consists of feedback from node about all the other nodes it has transacted with. Similarly, instead of single gossip weight $g_{ij}$, node will send vector $\bf{  g_i}$. A node id will also be attached with every pair of $y_{ij}$ and $g_{ij}$ so that receiving node can distinguish among  gossip pairs. So, in fact, node pushes gossip trio consisting of $y_{ij}$, $g_{ij}$ and node id $j$. The convergence of algorithm is checked by the following condition:
\begin{equation}
\sum\limits_{j}\bigg|\frac{y_{ij}(n)}{g_{ij}(n)}  -
\frac{y_{ij}(n-1)}{g_{ij}(n-1)}\bigg|   \leq N \xiup
\end{equation}

Where n is the time instant and $\xiup$ is the permissible error bound.

In the fourth variation we want to aggregate the globally calibrated local reputation of all the nodes simultaneously.  This algorithm is quite similar to second variation expect that we will use the third variation for gossiping process. Moreover in this variation nodes will push full vector $\bf{t_i}$, in place of only $t_{ij}$ for node $j$. 

It can be noted here that the time complexity of all four variations of algorithm will be of the same order because reputations of all the nodes will be pushed simultaneously as a vector. Whereas the communication complexity in third and fourth variation will increase proportionally to the size of vector, as now the reputation aggregation is happening for these many nodes. 
\subsection{Example for Differential Gossip Algorithm}
We will consider a network of 10 nodes and observe the aggregation in this network. Figure \ref{examplenet} shows the topology of the network. Table \ref{examplenet} shows the aggregated value after every iteration at each node. 
\begin{figure}[!t]
 \begin{center}
 \includegraphics[width=55mm,height=50mm, keepaspectratio=false]{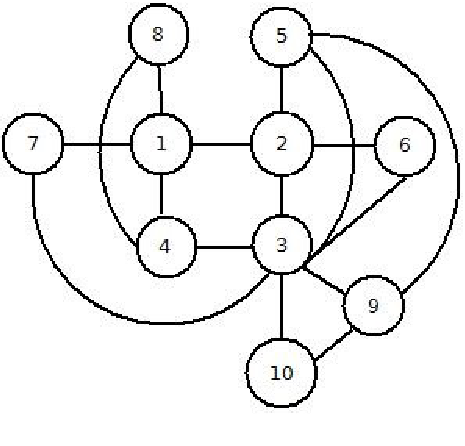}
 \caption{Topology of the example network}
 \label{examplenet}
 \end{center}
 \end{figure}
\begin{table*}[t]
\centering
\begin{tabular}{ |c|c| c| c| c|c|c|c|c|c|c|}
\hline
Node & 1& 2 & 3 & 4 & 5 & 6 & 7 & 8 & 9 & 10\\
\hline
degree & 4 & 4 & 7& 3 & 3 & 2 & 2 & 2 & 3 & 2\\
\hline
k      & 1 & 1 & 3 & 1 & 1 & 1 & 1 & 1 & 1 & 1 \\ 
itr=1  & 0.5653  & 0.3091  & 0.3629 & 0.4765 & 0.3080 & 0.6433 & 0.0668 & 0.6257 & 0.4386 & 0.7015\\
itr=2  & 0.4158  & 0.3091  & 0.4698 & 0.4765 & 0.3080 & 0.4923 & 0.2553 & 0.6257 & 0.4127 & 0.4860\\
itr=3  & 0.4158  & 0.3083  & 0.4073 & 0.5138 & 0.3080 & 0.4791 & 0.3898 & 0.5138& 0.4440 & 0.4324 \\
itr=4  & 0.3959  & 0.4128  & 0.4211 & 0.5138 & 0.3850 & 0.4791 & 0.3898 & 0.5138 & 0.4330 & 0.4324 \\ 
itr=5  & 0.4111  & 0.4301  & 0.4355 & 0.4222 & 0.3976 & 0.4388 & 0.3898 & 0.5138 & 0.4170 & 0.4329 \\
itr=6  & 0.4111  & 0.4301  & 0.4195 & 0.4287 & 0.4218 & 0.4388 & 0.3898 & 0.4287 & 0.4245 & 0.4338 \\
itr=7  & 0.4156  & 0.4246  & 0.4233 & 0.4248 & 0.4228 & 0.4388 & 0.4257 & 0.4287 & 0.4283 & 0.4283 \\
itr=8  & 0.4198  & 0.4244  & 0.4233 & 0.4187 & 0.4232 & 0.4262 & 0.4213 & 0.4269 & 0.4283 & 0.4270 \\
\hline
\end{tabular}
\caption{Aggregated value after every iteration at each node}
\label{tab:2}
\end{table*}

\section {Analysis of Algorithm}
\subsection{Analysis of convergence of Gossip Algorithm}
In this section we will study the time needed by nodes to converge to the average of local direct estimate values at different nodes. First, we will study the spreading of gossip in power law network. Then, we will study the diffusion speed of gossip. Then, we will study the diffusion speed of gossip. We have taken power law network as most actual P2P network tends to follow power law degree distribution. Based on these results, we will find the time of convergence.

Chierichetti \emph{et.al.} \cite{Flavio} proved that in PA based graph, $\{G_{N}^{m}\}$ for $m\geq 2$ push or pull alone will fail in spreading the gossip. They also proved that push-pull will succeed in $O((log_{2}N)^{2}$) time steps where low degree nodes push information to low degree nodes and power nodes and pull information from power node. 

But in a peer-to-peer networks, it's difficult to identify the power nodes. Moreover pulling the information is more expensive than pushing the information. So we have proposed to use differential push gossip in place of push pull gossip. In the following theorem we have proved that the differential push gossip will take same time as the push pull gossip.

\begin{theorem}
\label{spr}
Gossip will spread with high probability in a PA based graph, $\{G_{N}^{m}\}$ for $m\geq 2$, within $O((log_{2}N)^{2}$) time using differential-push (by high probability we mean, $1-o(1)$, where $o(1)$ goes to zero as N increases).
\end{theorem}

\begin{proof}
Differential push means that every node will push its data to different number ($k_{i}$) of nodes instead of one node. Here $k_{i}$ is the ratio of node's degree and average degree of all its neighbours.



$\{G_{N}^{m}\}$ for $m\geq 2$ can be thought as the union of few finite connected components. These components will consist of many low degree nodes with average degree $log_2 N$ and few high degree power nodes. These components will have diameter as $log_2 N$. 

In a network (component) with average degree $log_2 N$ and diameter $log_2 N$, the gossip will spread in within $O((log_{2}N)^{2})$ steps using normal push \cite{randomizedbroadcast}. But, these components also have power nodes that may lead to large spreading time in that component \cite{Flavio}. As power nodes will have high degree and making one push at a time will take longer in transferring the information to all of its neighbours. Therefore if some low degree node is connected only to power nodes, the information transfer to that node will take much longer. Differential push will solve this problem because now power nodes will be making multiple pushes as per the ratio of degrees of neighbour nodes and degree of the node itself. 

Hence, in every component, gossip spreads within $O((log_{2}N)^{2})$ steps using differential push. As there are finite such components, gossip will spread in complete network within $O((log_{2}N)^{2})$ steps.






\end{proof}
In our case, few nodes have information (reputation of a node) that has to be averaged and this average has to be spread to all the nodes. So we will prove the convergence for the case where every node has information which has to be averaged.

Let at $n=0$, each node has a number. For $jth$ node this number is $d_{0,j}$. So the objective of gossip is to have $\frac{\sum\limits_{j}d_{0,j}}{N}=d_{avg}$ at each node after some rounds of gossip. The number of rounds needed should be least possible. For $n=0$, the gossip weight at each node will be unity. After $n$ steps, let the node $j$ have the evolved number as $d_{n,j}$ and evolved gossip weight as $g_{n,j}$.

To study the time taken in the convergence of algorithm we assume that each node $j$ maintains a vector $\bf{c_{j}^{m}}$. The dimension of this vector will be $1\times N$. This vector will record the contribution received from every node including itself about the node $m$. So initially at $n=0$ each node will have a vector in which $(N-1)$ elements will be zero and one element, the one for itself, will be unity. If in the process of first step of  gossiping, only the node $i$ chooses node $j$ for pushing the gossip about $m$, then the contribution by $i$ to $j$ i.e. $c_{n,i,j}^{m}$ will be recorded in the $i^{th}$ element of contribution vector of node $j$. So after first step of gossip the contribution vector of $j$ will contain two non-zero elements one received from $i$ and one pushed to itself. We assume here that only push has been received by node $j$. In case of $l$ pushes being received the vector will have $l+1$ non zero entries.  Now node $j$  will choose some node $o$. Node $j$ will push the complete contribution vector divided by $p+1$(if $p$ push gossip is under consideration) to node $o$. Now node $o$ will do vector addition of all the received vectors including the one received from itself. The resultant vector will be the new contribution vector. This process will be repeated at all the nodes. 

So it can be said that $d_{n,j}^{m}=\sum\nolimits_{i} c_{n,i,j}^{m}\cdot d_{0,i}^{m}$ such that $g_{n,j}^{m}=\sum\nolimits_{i}c_{n,i,j}^{m}$. When a node will receive same amount of contribution from all nodes, at that time the ratio of evolved number ($d_{n,j}^{m}$) and evolved gossip weight ($g_{n,j}^{m}$) will be the average of all the numbers.

\begin{theorem}
\label{diff}
Uniform Gossip diffuses with differential push in PA based graph within $O((log_{2}N)^{2}+log_{2} \frac{1}{\xiup})$ time with high probability such that contributions at all nodes will be $\xiup$ uniform after this amount of time, i.e. $max_{i}|\frac{c_{n,i,j}^{m}}{||\bf c_{n,j}^{m}||_{1}}-\frac{1}{N}|\leq \xiup$ $\forall j$ where $||\bf c_{n,j}||_{1}=\sum\nolimits_{i}c_{n,i,j}$.
\end{theorem}

On the basis of these two theorems, it can be seen (as in \cite{Kempe}), that with high probability relative error in average estimation and sum estimation (if only one node is given weight one and others are given zero) will be bounded by $\xiup$ after $O((log_{2}N)^{2}+log_{2}\frac{1}{\xiup})$ gossip steps. 
\subsection{Analysis of Collusion}
In our proposed system a node may get trust values about a node by three possible ways, first by direct interaction, second from neighbours and third by gossiping. First mechanism can not be affected by collusion. We are assuming that second mechanism will also not be affected by collusion as neighbours have a definite level of trust for each other. We are considering the collusion because of third mechanism.

For analysis of collusion, we will calculate the difference of real reputation  and estimated reputation of a node $j$ by some node $o$ in the presence of collusion using our proposed method, we will compare it with the method proposed in \cite{gossiptrust}. Lets us assume that the network is formed by the member nodes of set $\bf{\mathbb{N}}$. There is a subset $\bf{\mathbb{C}}$ of set $\bf{\mathbb{N}}$ such that member nodes of set $\bf{\mathbb{C}}$ are involved in collusion. The cardinality of sets $\bf{\mathbb{N}}$ and $\bf{\mathbb{N}}$ are assumed to be $N$ and $C$ respectively. We also assume that members nodes of set $\bf{\mathbb{C}}$ are colluding in groups with a group size of $G$. By colluding in a group we mean that if some node is the member of that group then group members of colluding group will report its reputation as $1$. Whereas for others nodes they will report the reputation value as $0$. Let us say that real reputation of a node $j$ is $R_{real}^{j}$, and estimated reputation is $R_{estc}^{j}$, if $j$ is a colluding node, $R_{estnc}^{j}$, if $j$ is not a colluding node.
\begin{equation}
R_{real}^{j}=\frac{\sum\limits_{i\in \bf{\mathbb{N}}}t_{ij}}{N}.
\end{equation}
Here $t_{ij}$ is normalised trust value of the node $x$ at node $i$.\\
If $j$ is not a colluding node then,
\begin{equation}
R_{estnc}^{j}=\frac{\sum\limits_{i\in \bf{\mathbb{N}}\setminus\bf{\mathbb{C}}}t_{ij}}{N}.
\end{equation}
If $j$ is a colluding node, then
\begin{equation}
R_{estc}^{j}=\frac{\sum\limits_{i\in \bf{\mathbb{N}}\setminus\bf{\mathbb{C}}}t_{ij}+G}{N}.
\end{equation}
So the expected value of reputation estimate ($E[R_{est}^{j}]$) will be
\begin{eqnarray}
\nonumber E[R_{est}^{j}]&=& \frac{C}{N}\left(\frac{\sum\limits_{i\in \bf{\mathbb{N}}\setminus\bf{\mathbb{C}}}t_{ij}+G}{N}\right)+ \left(1-\frac{C}{N}\right)\left(\frac{\sum\limits_{i\in \bf{\mathbb{N}}\setminus\bf{\mathbb{C}}}t_{ij}}{N}\right)\\
&=& \frac{GC}{N^{2}}+\frac{\sum\limits_{i\in \bf{\mathbb{N}}\setminus\bf{\mathbb{C}}}t_{ij}}{N}
\end{eqnarray}
So difference in real reputation and expected value of estimated reputation by node $o$ for node $j$ ($\Delta R_{old}^{oj}$) will be,
\begin{equation}
\Delta R_{old}^{oj}=-\frac{GC}{N^{2}}+\frac{\sum\limits_{i\in \bf{\mathbb{C}}}t_{ij}}{N}.
\end{equation}
Now, we incorporate the trust based weighted opinion of neighbours. Lets us assume that $w_{oi}$ is the weight given to the opinion of node $i$ by node $o$. It may be noted $w_{oi}\geq 1, \forall i\; (equation ~\ref{weight})$. So the real reputation of node $j$ for node $o$ will be,
\begin{equation}
R_{real}^{j}=\frac{\sum\limits_{i\in \bf{\mathbb{N}}}t_{ij}+\sum\limits_{i\in \bf{\mathbb{N}}}(w_{oi}-1)t_{ij}}{N+\sum\limits_{i\in \bf{\mathbb{N}}}(w_{oi}-1)}.
\end{equation}
If $x$ is not a colluding node then,
\begin{equation}
R_{estnc}^{j}=\frac{\sum\limits_{i\in \bf{\mathbb{N}}\setminus\bf{\mathbb{C}}}t_{ij}+\sum\limits_{i\in \bf{\mathbb{N}}}(w_{oi}-1)t_{ij}}{N+\sum\limits_{{i\in \bf{\mathbb{N}}}}(w_{oi}-1)}.
\end{equation}
And if $j$ is a colluding node then,
\begin{equation}
R_{estc}^{j}=\frac{\sum\limits_{i\in \bf{\mathbb{N}}\setminus\bf{\mathbb{C}}}t_{ij}+\sum\limits_{i\in \bf{\mathbb{N}}}(w_{oi}-1)t_{ij}+G}{N+\sum\limits_{i\in \bf{\mathbb{N}}}(w_{oi}-1)}.
\end{equation}
So the expected value of reputation estimate ($E[R_{est}^{j}]$) will be
\begin{eqnarray}
\nonumber E[R_{est}^{j}]&=& \frac{C}{N}\left(\frac{\sum\limits_{i\in \bf{\mathbb{N}}\setminus\bf{\mathbb{C}}}t_{ij}+\sum\limits_{i\in \bf{\mathbb{N}}}(w_{oi}-1)t_{ij}+G}{N+\sum\limits_{i\in \bf{\mathbb{N}}}(w_{oi}-1)}\right)\\\nonumber &&+ \left(1-\frac{C}{N}\right)\left(\frac{\sum\limits_{i\in \bf{\mathbb{N}}\setminus\bf{\mathbb{C}}}t_{ij}+\sum\limits_{i\in \bf{\mathbb{N}}}(w_{oi}-1)t_{ij}}{N+\sum\limits_{i\in \bf{\mathbb{N}}}(w_{oi}-1)}\right)\\
\nonumber&=&\frac{GC}{N\cdot (N+\sum\limits_{i\in \bf{\mathbb{N}}}(w_{oi}-1))}\\&&+\left(\frac{\sum\limits_{i\in \bf{\mathbb{N}}\setminus\bf{\mathbb{C}}}t_{ij}+\sum\limits_{i\in \bf{\mathbb{N}}}(w_{oi}-1)t_{ij}}{N+\sum\limits_{i\in \bf{\mathbb{N}}}(w_{oi}-1)}\right).
\end{eqnarray}
So difference in real reputation and expected value of estimated reputation by node $o$ for node $j$ ($\Delta R_{new}^{oj}$)will be  (equation 16 - equation 13),
\begin{eqnarray}
\nonumber\Delta R_{new}^{oj}&=&-\frac{GC}{N\cdot (N+\sum\limits_{i\in \bf{\mathbb{N}}}(w_{oi}-1))}+\frac{\sum\limits_{i\in \bf{\mathbb{C}}}t_{ij}}{N+\sum\limits_{i\in \bf{\mathbb{N}}}(w_{oi}-1)}\\
\nonumber&=&\frac{N}{(N+\sum\limits_{i\in \bf{\mathbb{N}}}(w_{oi}-1))}\cdot \left(-\frac{GC}{N^{2}}+\frac{\sum\limits_{i\in \bf{\mathbb{C}}}t_{ij}}{N}\right)\\
&=&\frac{N}{(N+\sum\limits_{i\in \bf{\mathbb{N}}}(w_{oi}-1))}\cdot \Delta R_{old}^{oj}
\end{eqnarray}

\subsection{Numerical Results}
Performance of algorithm for reputation aggregation for peer to peer file sharing system is also evaluated by simulation as well. 
The simulation experiments has been conducted for 100 to 50000 nodes. A power law network has been built using
\emph{Preferential Attachment} model. Performance of differential algorithm has been evaluated in terms of number of iterations (to assess the rate of convergence) required to converge within a certain aggregation error. Number of packets per node per gossip step that are required to be transmitted for convergence have also been calculated to assess the network overhead. Algorithm has also been tested against collusion. 

Figure ~\ref{fig:gossipwithoutloss} shows the number of gossip steps required for different error bounds for different number of nodes. This is clearly evident that number of gossip steps is increasing with a rate much less than normal push gossip.

 Peer to peer networks operate above TCP layer, i.e. these kind of networks assume a reliable bit pipe between sender and receiver. So peer to peer network suffers by packet loss only when some node leaves the network i.e. due to churning.Figure ~\ref{fig:N=10000packetloss} shows the required number of gossip steps with different packet loss probability for 10000 nodes. Here the assumption is when a node leaves during gossip process, it hands over the gossip pair vectors to some other node so mass conservation still applies. Whenever a node pushes gossip pair to this absent node, the pushing node doesn't receive any acknowledgement. In such cases pushing node pushes the gossip pair to itself so that mass conservation still applies. 
 We can see a small increment in the number of gossip steps with the increase in the packet loss probability.

Figure ~\ref{fig:collgroup} and ~\ref{fig:collindi} shows the immunity of algorithm against collusion in terms of RMS error in case of individual (fig~\ref{fig:collindi}) and group collusion (fig~\ref{fig:collgroup}). 
Here average RMS error is defined as follows.
\begin{equation}
Average\; RMS\; error= \frac{1}{N}\sum\limits_{i}\sqrt{\frac{\sum\limits_{j}((r_{ij}-\hat{r}_{ij})/r_{ij})^{2}}{N}}
\end{equation}
Here $r_{ij}$ is the reputation of node $i$ at node $j$ computed by differential gossip in presence of colluding nodes, whereas $\hat{r_{ij}}$ is the computed reputation  if colluding nodes would not have been there. This is clearly evident that effect of collusion on reputation computation by differential gossip is quite less even with very high percentage of colluding users. The colluding group size is making a small difference in differential gossip reputation computation.

Table ~\ref{tab:1} shows the number of message transfers required by a node in one gossip step. It can be seen that this is decreasing slightly with the increase in number of nodes. This is happening because as number of gossip steps increases the overhead incurred in the beginning gets distributed and a node is less burdened as the number of total nodes increases. Similar thing happens when a lower value of $\xiup$ is chosen. This is also evident from Table ~\ref{tab:1} and fig~\ref{fig:gossipwithoutloss} that in case of differential gossip per step communication cost is more than normal push gossip but total communication cost for convergence is less for networks bigger than 1000 nodes moreover this differences increases substantially as network size increases. We have not verified it for normal pull gossip but intuitively it can be observed that same thing will be true.
\begin{figure}[!t]
\begin{center}
\includegraphics[width=80mm, height=70mm, keepaspectratio=false]{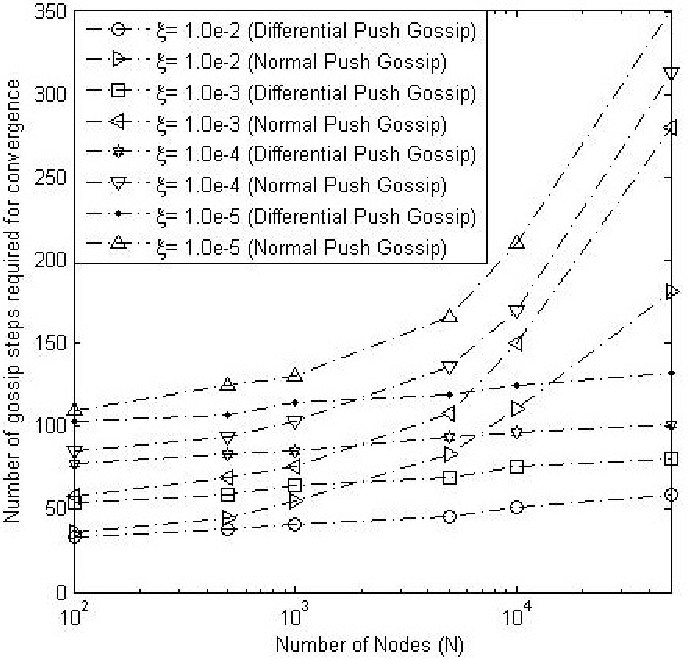}
\caption{Gossip step counts with different number of nodes(N) and different error bounds$\xiup$}
\label{fig:gossipwithoutloss}
\end{center}
\end{figure}

\begin{figure}[!t]
\begin{center}
\includegraphics[width=80mm, height=70mm, keepaspectratio=false]{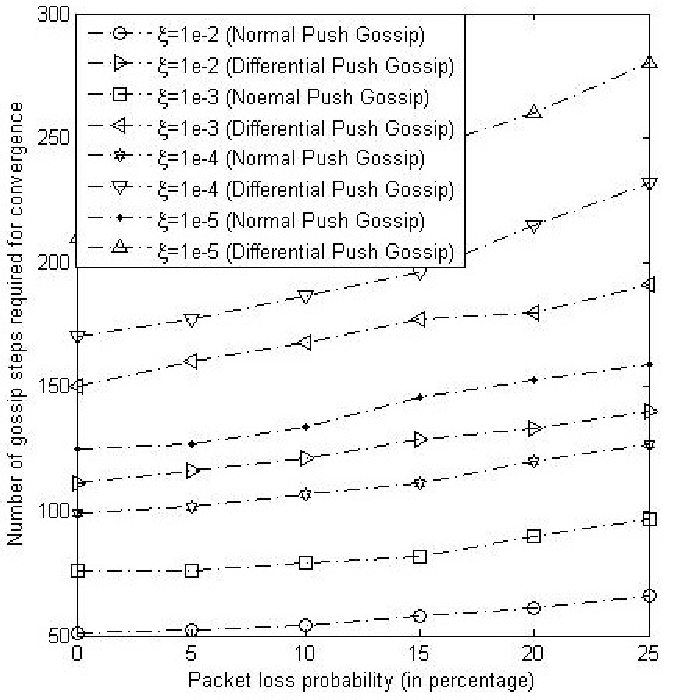}
\caption{Gossip step counts for N=10000 with different error bounds$\xiup$ for different packet loss probability}
\label{fig:N=10000packetloss}
\end{center}
\end{figure}


\begin{figure}[!t]
\begin{center}
\includegraphics[width=80mm,height=70mm, keepaspectratio=false]{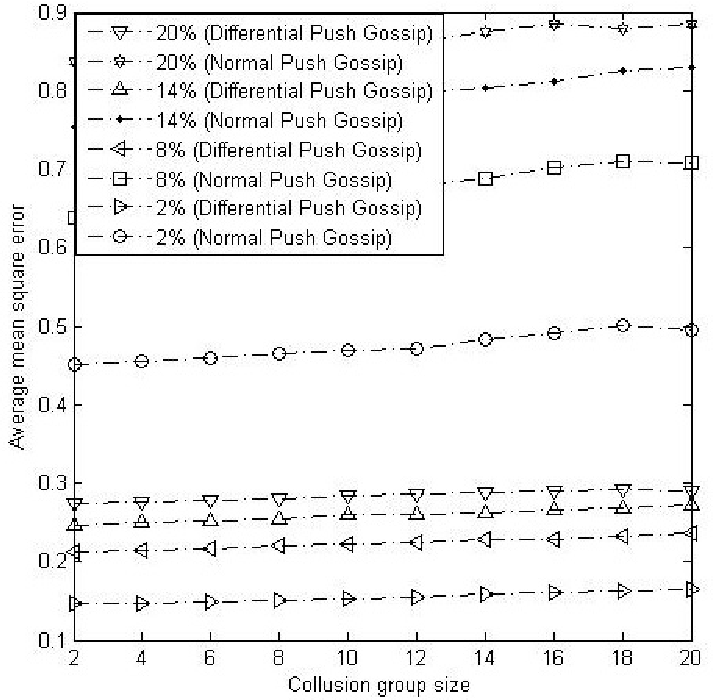}
\caption{Average RMS error with different size colluding groups for different percentage of colluding peers}
\label{fig:collgroup}
\end{center}
\end{figure}

\begin{figure}[!t]
\begin{center}
\includegraphics[width=80mm,height=75mm, keepaspectratio=false]{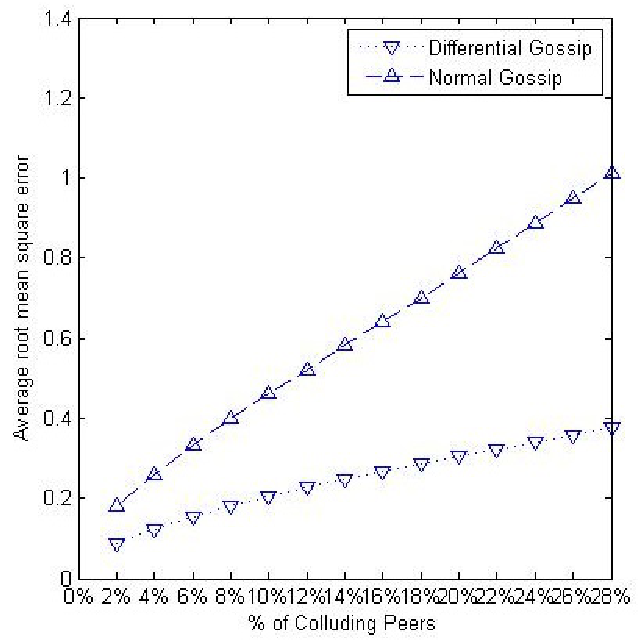}
\caption{Average RMS error with individual peers  for different percentage of colluding peers}
\label{fig:collindi}
\end{center}
\end{figure}
\begin{table}[t]
\centering
\begin{tabular}{ |c|c| c| c| c|}
\hline
 &$\xiup$=0.01 &$\xiup=0.001$ & $\xiup=0.0001$ & $\xiup=0.00001$ \\
\hline
\bf{N=100}    & 1.212    & 1.203   & 1.195 & 1.188  \\
\bf{N=500}    & 1.199    & 1.194   & 1.189 & 1.183  \\
\bf{N=1000}   & 1.178    & 1.159   & 1.157 & 1.148  \\
\bf{N=10000}  & 1.156    & 1.139   & 1.124 & 1.122\\
\bf{N=50000}  & 1.152    & 1.132   & 1.119 & 1.112\\
\hline
\end{tabular}
\caption{Number of messages per node per step transmitted due to gossiping}
\label{tab:1}
\end{table}
\section{Conclusion}
In peer-to-peer networks, free riding is a major problem that can be overcome by using reputation management system. A reputation management system includes two processes, first estimation of reputation and second aggregation of reputation. In this paper we have proposed an aggregation technique by modifying push gossip algorithm to differential push gossip algorithm. 

The proposed aggregation technique efficiently aggregates the trust values from different nodes in a power law network. This technique does not require the identification of power nodes. This makes algorithm easily implementable as identification of power nodes in a distributed setting is hard. This algorithm is also robust against churning as can be seen in figure~\ref{fig:N=10000packetloss}.
  Proposed technique aggregates the reputation in a differential manner. This is done by considering the feedback of trusted nodes with a higher weight. This leads to robustness against collusion as evident from figure~\ref{fig:collgroup}. 

Proposed algorithm has been presented to avoid the problem of free riding but it can also be used to avoid malicious users in the network just by changing the method of estimation of $a_i$ and $b_{ij}$.
\bibliography{ref}
\bibliographystyle{ieeetran}
\appendix[proof of theorem~\ref{diff}]
We can see the property of \emph{mass conservation} (proposition ~\ref{massconv}) \cite{Kempe} holds in this case as well.
\begin{proposition}
\label{massconv}
Under the differential Push protocol with Uniform Gossip, the sum
of all of $i^{th}$ node's contributions at all nodes j is $\sum\nolimits_{j}c_{n,i,j}^{m}=1$ and hence the sum of all weights is $\sum\nolimits_{j}g_{n,j}^{m}=N$
\end{proposition}
\begin{proof}[Proof of theorem~\ref{diff}]
As we know when a node will get equal contribution from every node it will reach the average value. Taking variance around the mean value of the contributions from all the nodes at a particular node will give the level of convergence at one node. If we sum these variances for all the nodes, we will get the idea about the convergence of network. We are just referring to reputation of node $m$ in this proof, and super script $m$ have not been explicitly shown. It means $c_{nij}^{m}=c_{nij}$ and $g_{nj}^{m}=g_{nj}$. The variance at node j will be
$E(c_{n,i,j}-\overline{c_{n,i,j}})^{2}$ i.e. $\frac{1}{N}\sum\limits_{i}(c_{n,i,j}-\frac{g_{n,j}}{N})^{2}$. As this quantity is small, we can drop $\frac{1}{N}$. This will still give the idea about the convergence. Further adding the variance at all the nodes gives us the idea about further convergence in the whole network. We call this as potential function $\psiup_{n}$
\begin{equation}
\psiup_{n}=\sum_{j,i}\left(c_{n,i,j}-\frac{g_{n,j}}{N}\right)^2
\end{equation}
Let us study $\psiup$ for p-push gossip, i.e. when every node is making p pushes to p nodes. Here we are assuming that node chooses every node including itself for push independently. Here f(k)=j means  that a node k chooses a node j and pushes gossip pair. The $d_{n,j}$ and $g_{n,j}$ are divided by $(p+1)$; one part is always retained by the node and remaining are used for $p$ push.
\begin{eqnarray}
c_{n+1,i,j}&=&\frac{1}{p+1}c_{n,i,j}+\frac{1}{p+1}\sum\limits_{k:f(k)=j}c_{n,i,k}.\\
g_{n+1,j}&=&\sum\limits_{i}c_{n+1,i,j}\\
\nonumber&=&\sum\limits_{i}\frac{1}{p+1}c_{n,i,j}+\sum\limits_{i}\frac{1}{p+1}\sum\limits_{k:f(k)=j}c_{n,i,k}\\
\nonumber&=&\frac{1}{p+1}g_{n,j}+\frac{1}{p+1}\sum\limits_{k:f(k)=j}g_{n,k}\\
\nonumber As\\
\psiup_{n+1}&=&\sum_{j,i}\left(c_{n+1,i,j}-\frac{g_{n+1,j}}{N}\right)^2, 
\end{eqnarray}
substituting the values of $c_{n+1,i,j}$ and $g_{n+1,j}$, we get
\begin{eqnarray}
\nonumber \psiup_{n+1}&=&\sum\limits_{j,i}\left(\frac{1}{p+1}\left(c_{n,i,j}-\frac{g_{n,j}}{N}\right)+ \sum\limits_{k:f(k)=j}\frac{1}{p+1}\left(c_{n,i,k}-\frac{g_{n,k}}{N}\right)\right)^2\\
\nonumber &=& \frac{1}{(p+1)^{2}}\sum\limits_{j,i}\left(c_{n,i,j}-\frac{g_{n,j}}{N}\right)^2\\ \nonumber &&+ \frac{1}{(p+1)^{2}}\sum\limits_{j,i}\sum\limits_{k:f(k)=j}\left(c_{n,i,k}- \frac{g_{n,k}}{N}\right)^2\\
\nonumber &&+ \frac{2}{(p+1)^{2}}\sum\limits_{j,i}\sum\limits_{k:f(k)=j}\left(c_{n,i,j}-\frac{g_{n,j}}{N}\right) \left(c_{n,i,k}-\frac{g_{n,k}}{N}\right)\\
\nonumber &&+ \frac{1}{(p+1)^{2}}\sum\limits_{j,i}\sum\limits_{\hat{k}}\sum\limits_{k:f(k)=f(\hat{k})=j,k\neq \hat {k}}\left(c_{n,i,k}-\frac{g_{n,k}}{N}\right)\\\nonumber && \left(c_{n,j,\hat{k}}-\frac{g_{n,\hat{k}}}{N}\right).
\end{eqnarray}
In second term of the previous equation, we are doing summation over $i$, $j$ and $k:f(k)=j$. Each node $k$ is contributing $p$ copies to its $p$ neighbors. If this contribution is summed over $k$, it should be equal to contribution received by each node $j$ when summed over all nodes $j$. Thus \begin{equation}
\sum_k p ( c_{n,i,k} - { g_{n,k} \over N } )^2 = \sum_j \sum_{k: f(k)=j} (c_{n,i,k} - { g_{n,k} \over N } )^2
\end{equation}.
\begin{eqnarray}
\nonumber\psiup_{n+1} &=& \frac{1}{(p+1)^{2}}\sum\limits_{j,i}\left(c_{n,i,j}-\frac{g_{n,j}}{N}\right)^2+ \frac{p}{(p+1)^{2}}\sum\limits_{j,i}\left(c_{n,i,j}- \frac{g_{n,j}}{N}\right)^2\\
\nonumber &&+ \frac{2}{(p+1)^{2}}\sum\limits_{j,i}\sum\limits_{k:f(k)=j}\left(c_{n,i,j}-\frac{g_{n,j}}{N}\right) \left(c_{n,i,k}-\frac{g_{n,k}}{N}\right)\\
\nonumber &&+ \frac{2}{(p+1)^{2}}\sum\limits_{j,i}\sum\limits_{k\neq \hat {k}:f(k)=f(\hat{k})=j}\left(c_{n,i,k}-\frac{g_{n,k}}{N}\right) \left(c_{n,j,\hat{k}}-\frac{g_{n,\hat{k}}}{N}\right)\\
\nonumber &=& \frac{1}{(p+1)}\psiup_{n}+\frac{2}{(p+1)^{2}}\\\nonumber&&\sum\limits_{j,i,k:f(k)=j} \left(c_{n,i,j}-\frac{g_{n,j}}{N}\right) \left(c_{n,i,k}-\frac{g_{n,k}}{N}\right)+ \frac{2}{(p+1)^{2}}\\ 
\nonumber && \sum\limits_{j,i}\sum\limits_{\hat{k}}\sum\limits_{k:f(k)=f(\hat{k}),\hat{k}\neq k}\left(c_{n,i,k}-\frac{g_{n,k}}{N}\right) \left(c_{n,i,\hat{k}}-\frac{g_{n,\hat{k}}}{N}\right).
\end{eqnarray}
We know that node will choose a node randomly among its neighbours. Let us assume that the degree of said node $k$ is $d_{k}$ with probability $P_{d_{k}}$. Out of remaining nodes, we can form group of $d_k$ nodes in $N-1\choose d_{k}$ ways. If we fix one node say $j$ as one of the neighbour, then there are ${N-2} \choose C_{d_k -1}$ ways of having other $d_k -1$ neighbours. Thus the probability of j being a neighbour of a node k having degree $d_k$ will $ { {{N-2} \choose{d_k -1}} \over { {N-1} \choose {d_k}}}$. Further the probability that $k$ will choose $j$ will be $ 1 \over d_k $.
\begin{eqnarray}
\nonumber P[f(k)=j]&=&\frac{{N-2\choose d_{k}-1}}{{N-1\choose d_{k}}}\cdot P_{d_{k}}\cdot \frac{1}{{d_{k}}}   \\
 &=& \frac{P_{d_{k}}}{N-1}\\
Similarly\\
 P[\hat{k}\neq k,f(k) = f(\hat{k})=j] &=& \frac{P_{d_{k}}}{N-1}\cdot \frac{P_{d_{\hat{k}}}}{N-1}
\end{eqnarray}
The $P_{d_{k}}$ is the probability that a node has degree $d_{k}$. For networks generated by PA Model, 
\begin{equation}
\nonumber P(d_{k})\sim d_{k}^{-\gamma}.
\end{equation} 
Here $\gamma$ is network exponent. Hence,
\begin{eqnarray}
\nonumber  E[\psiup_{n+1}\mid \psiup_{n}] &=& \frac{1}{(p+1)}\psiup_{n}+\frac{2}{(p+1)^{2}} \sum\limits_{j,k,i}\left(c_{n,i,j}-\frac{g_{n,j}}{N}\right) \\ \nonumber &&\left(c_{n,i,k}-\frac{g_{n,k}}{N}\right)P[f(k) = j]\\
\nonumber &&+ \frac{2}{(p+1)^{2}}\sum\limits_{i,j}\sum\limits_{\hat{k}}\sum\limits_{k:\hat{k}\neq k}\left(c_{n,i,k}-\frac{g_{n,k}}{N}\right) \\ \nonumber && \left(c_{n,i,\hat{k}}-\frac{g_{n,\hat{k}}}{N}\right)P[\hat{k}\neq k,f(k) = f(\hat{k})=j]\\
\nonumber &=& \frac{1}{p+1}\psiup_{n}+\frac{2}{(p+1)^{2}} \sum\limits_{j,k,i}\left(c_{n,i,j}-\frac{g_{n,j}}{N}\right)\\ \nonumber && \left(c_{n,i,k}-\frac{g_{n,k}}{N}\right)\frac{P_{d_{k}}}{N-1}
+\frac{N}{(p+1)^{2}}\\ \nonumber && \sum\limits_{i}\sum\limits_{\hat{k}}\sum\limits_{k}\left(c_{n,i,k}-\frac{g_{n,k}}{N}\right)  \left(c_{n,i,\hat{k}}-\frac{g_{n,\hat{k}}}{N}\right) \frac{P_{d_{k}}}{N-1}\\ 
\nonumber &&\cdot \frac {P_{d_{\hat{k}}}}{N-1}- \frac{N}{(p+1)^{2}}\sum\limits_{i,k}\left(c_{n,i,k}-\frac{g_{n,k}}{N}\right)^2 \frac{P_{d_{k}}^{2}}{(N-1)^{2}}\\
\nonumber &=& \frac{1}{p+1}\psiup_{n}+\frac{2}{(p+1)^{2}(N-1)} \\ \nonumber &&\sum\limits_{i}\sum\limits_{j}\left(c_{n,i,j}-\frac{g_{n,j}}{N}\right) \sum\limits_{k} \left(c_{n,i,k}-\frac{g_{n,k}}{N}\right)P_{d_{k}}\\
\nonumber &&+\frac{N}{(p+1)^{2}(N-1)^2}\sum\limits_{i}\sum\limits_{k}\left(c_{n,i,k}-\frac{g_{n,k}}{N}\right)P_{d_{k}}\\ \nonumber && \sum\limits_{\hat{k}}\left(c_{n,i,\hat{k}}-\frac{g_{n,\hat{k}}}{N}\right) P_{d_{\hat{k}}}\\ 
\nonumber &&-\frac{N}{(p+1)^{2}(N-1)^2}\sum\limits_{k,i}\left(c_{n,i,k}-\frac{g_{n,k}}{N}\right)^2 P_{d_{k}}^{2}
\end{eqnarray}
Let us assume that the maximum value of $P_{d}$ is $P_{d_{max}}$ and minimum value is $P_{d_{min}}$ and the difference of $P_{d_{max}}$ and $P_{d_{min}}$ is $K_{c}$, then
\begin{eqnarray}
\nonumber  E[\psiup_{n+1}\mid \psiup_{n}] &\leq & \frac{1}{p+1}\psiup_{n}\\ \nonumber &&+\frac{2}{(p+1)^{2}(N-1)} \sum\limits_{i}\left(\sum\limits_{j}c_{n,i,j}-\sum\limits_{j}\frac{g_{n,j}}{N}\right)\\ \nonumber && \sum\limits_{k} \left(c_{n,i,k}-\frac{g_{n,k}}{N}\right)P_{d_{k}}\\
\nonumber &&+ \frac{N}{(p+1)^{2}(N-1)^2}\sum\limits_{i}\left(\sum\limits_{k}c_{n,i,k}P_{d_{max}}-
\sum\limits_{k}\frac{g_{n,k}}{N}P_{d_{min}}\right)\\ \nonumber && \left(\sum\limits_{\hat{k}}c_{n,i,\hat{k}}P_{d_{max}}-
\sum\limits_{\hat{k}}\frac{g_{n,\hat{k}}}{N}P_{d_{min}}\right)\\ 
\nonumber &&-\frac{N}{(p+1)^{2}(N-1)^2}\sum\limits_{i,k}\left(c_{n,i,k}-\frac{g_{n,k}}{N}\right)^2 P_{d_{min}}^{2}
\end{eqnarray}
Applying mass conservation $\sum\limits_{j}c_{n,i,j}$, $\sum\limits_{k}c_{n,i,k}$, $\sum\limits_{j}\frac{g_{n,j}}{N}$ and $\sum\limits_{k}\frac{g_{n,k}}{N}$ will be equal to $1$,  so the second term will become zero and third term will become $N\cdot K_{c}^{2}$. Fourth term is always non negative so removing this term will not affect the bound. So
\begin{eqnarray}
\nonumber E[\psiup_{n+1}\mid \psiup_{n}] &\leq& \frac{1}{p+1}\psiup_{n}+\frac{1}{(p+1)^{2}(N-1)^2}\cdot N^2 \cdot K_{c}^{2}\\ \nonumber &&- \frac{N}{(p+1)^{2}(N-1)^2}P_{d_{k}}^{2} \sum\limits_{k,i}\left(c_{n,i,k}-\frac{g_{n,k}}{N}\right)^2\\
\nonumber &\leq & \frac{1}{p+1}\psiup_{n}+\frac{ K_{c}^{2}N^2}{(p+1)^{2}(N-1)^2} \\\nonumber&& - \frac{N}{(p+1)^{2}(N-1)^2}P_{d_{min}}^{2} \psiup_{n}
 \\\nonumber&\leq  & \frac{1}{p+1}\psiup_{n}+\frac{K_{c}^{2}N^2}{(p+1)^{2}(N-1)^2}
\end{eqnarray}
In the last line we use the fact that $\frac{N}{(p+1)^{2}(N-1)^2}P_{d_{min}}^{2} \psiup_{n}$ will always remain non negative. 
 We know that $d_{min}$ is $2$. If we consider the value of $\gamma$ to be 2, maximum value of $ K_{c}^{2}$ will be $\frac{1}{16}$ considering $P_{d_{max}}$ to be zero and the maximum possible value of $\frac{N^2}{(N-1)^2}$ is 4. Thus,  
\begin{equation}
\label{finalcondfi}
E[\psiup_{n+1}\mid \psiup_{n}] \leq \frac{1}{p+1}\psiup_{n}+\frac{1}{4\cdot (p+1)^{2}}.
\end{equation}
Now we will calculate the value of $\psiup_{0}$. We know that initially the contribution vector $\bf{c_{j}}$ contains only single non-zero value i.e. contribution received from it self and that value is $1$, rest all $N-1$ elements are $0$. So
\begin{eqnarray}
\nonumber\psiup_{0}&=&\sum\limits_{j,i}\left(c_{0,i,j}-\frac{g_{0,j}}{N}\right)^2\\
\nonumber &=& \sum\limits_{j}[\left(c_{0,1,j}-\frac{g_{0,j}}{N}\right)^2+ \left(c_{0,2,j}-\frac{g_{0,j}}{N}\right)^2+...+\\ \nonumber &&\left(c_{0,j,j}-\frac{g_{0,j}}{N}\right)^2+...+ \left(c_{0,N,j}-\frac{g_{0,j}}{N}\right)^2]\\
\nonumber &=&\sum\limits_{j}[\left(0-\frac{1}{N}\right)^2+ \left(0-\frac{1}{N}\right)^2+...+\\ \nonumber &&\left(1-\frac{1}{N}\right)^2+...+ \left(0-\frac{1}{N}\right)^2]\\
\nonumber &=& \sum\limits_{j}[1+\frac{1}{N^{2}}-\frac{2}{N}+(N-1)\cdot \frac{1}{N^{2}}]\\
 &=& \sum\limits_{j}[1-\frac{1}{N}]=N-1
\end{eqnarray}
Now we will substitute the value of $\psiup_{0}$ in (~\ref{finalcondfi}). This will give us the bound on $\psiup_{n}$
\begin{eqnarray}
E[\psiup_{1}\mid \psiup_{0}] &\leq& \frac{1}{p+1}\psiup_{0}+\frac{1}{4\cdot (p+1)^{2}}\\
\nonumber E[\psiup_{1}] &\leq& \frac{1}{p+1}(N-1)+\frac{1}{4\cdot (p+1)^{2}}\\
Similarly\\
\nonumber E[\psiup_{2}] &\leq& \frac{1}{p+1}(\frac{1}{p+1}(N-1)+\frac{1}{4\cdot (p+1)^{2}})\\
\nonumber&&+\frac{1}{16\cdot (p+1)^{2}}\\
\nonumber&=& \frac{N-1}{(p+1)^{2}}+\frac{1}{16(p+1)^{3}}+\frac{1}{4(p+1)^{2}}\\
\nonumber E[\psiup_{n}]&\leq& \frac{N-1}{(p+1)^{n}}+\frac{1}{16(p+1)^{n+1}}\\\nonumber&& +\frac{1}{4(p+1)^{n}}+...+\frac{1}{4(p+1)^{2}}\\
\nonumber &=& (N-1) \cdot (p+1)^{-n} +\frac{1}{4\cdot (p+1)p}\\\nonumber&&-\frac{1}{4p(p+1)^{n+1}}\\
\nonumber & \leq & (N-1) \cdot (p+1)^{-n} +\frac{1}{4\cdot(p+1)p}
\end{eqnarray}
It can be seen that right hand side of the above equation is maximum when p=1. It means potential function is decaying at the slowest rate for p=1. So time taken in convergence for p=1 will be maximum. For normal push, algorithm will act as upper bound for differential push algorithm i.e. combination of different values of positive integer p's. So taking p=1; 
\begin{eqnarray}
\nonumber E[\psiup_{n}]&\leq& (N-1) \cdot 2^{-n} + \frac{1}{8}\\
&\leq & (N-1) \cdot 2^{-n} \cdot k_{d}
\end{eqnarray}
Here $k_{d}$ is an integer constant that is greater than the ratio of maximum value of $(N-1) \cdot 2^{-n}$ and $\frac{1}{8}$. This can be seen that $k_{d}$ will depend on the number of steps required for convergence. 

After gossiping for $n=log_{2}(N-1)+log_{2}k_{d}+log_{2}\frac{1}{\xiup}$ steps,
\begin{eqnarray}
\nonumber E[\psiup_{n}]&\leq& (N-1)\cdot 2^{-(log_{2}(N-1))}\cdot 2^{-(log_{2}(k_{d}))}\\\nonumber&&\cdot2^{-(log_{2}(1/\xiup))}\cdot k_{d}\\
E[\psiup_{n}]&\leq& \xiup\\
\nonumber \sum_{j,i}\left(c_{n,i,j}-\frac{g_{n,j}}{N}\right)^2 &\leq& \xiup
\end{eqnarray}
If summation of some non-negative numbers are less than $\xiup$ then individually each number must be less than $\xiup$, i.e. $|c_{n,i,j}-\frac{g_{n,j}}{N}|\leq \xiup^{\frac{1}{2}} $ for all nodes $i$.

If we consider weight as an information to be spread among the nodes, according to theorem~\ref{spr}, information will reach to all the nodes in a power law network with high probability, in $n=(log_{2}N)^{2}$ rounds. After these n rounds every node will receive at least $2^{-n}$ weight. So applying union bound \cite{UB} over weight spreading and potential decay event (We have seen potential is decaying at every step) and dividing with $g_{n,j}$ gives $|\frac{c_{n,i,j}}{g_{n,j}}-\frac{1}{N}|\leq \xiup$  at steps $O((log_{2}N)^{2}+log_{2}N+log_{2}k_{d}+log_{2}\frac{1}{\xiup})$ i.e. within $O((log_{2}N)^{2}+log \frac{1}{\xiup})$ steps with high probability.
\end{proof}

\end{document}